\documentclass[journal,comsoc]{IEEEtran}
\usepackage[T1]{fontenc}
\usepackage{tabularx}
\usepackage{xcolor}
\usepackage{multirow}
\usepackage{booktabs}
\usepackage{array}
\newcolumntype{L}[1]{>{\raggedright\let\newline\\\arraybackslash\hspace{0pt}}m{#1}}
\newcolumntype{C}[1]{>{\centering\let\newline\\\arraybackslash\hspace{0pt}}m{#1}}
\newcolumntype{R}[1]{>{\raggedleft\let\newline\\\arraybackslash\hspace{0pt}}m{#1}}
\usepackage[english]{babel}
\usepackage{amsthm}
\usepackage{amssymb}
\newtheorem{theorem}{Theorem}
\newtheorem{lemma}{Lemma}
\theoremstyle{plain}
\newtheorem{corollary}{Corollary}
\theoremstyle{plain}

\theoremstyle{remark}

\usepackage{cite}
\ifCLASSINFOpdf
   \usepackage[pdftex]{graphicx}
\else
   \usepackage[dvips]{graphicx}
   \graphicspath{{../eps/}}
   \DeclareGraphicsExtensions{.eps}
\fi
\usepackage[normalem]{ulem}
\usepackage{amsmath}
\usepackage[cmintegrals]{newtxmath}
\usepackage{algorithmic}
\usepackage{multicol}
\usepackage{algorithm}
\usepackage{verbatim}
\usepackage{array}
\usepackage{makecell}
\usepackage[pagewise]{lineno}
\usepackage{upgreek}
\ifCLASSOPTIONcompsoc
  \usepackage[caption=false,font=normalsize,uabelfont=sf,textfont=sf]{subfig}
\else
  \usepackage[caption=false,font=footnotesize]{subfig}
\fi
\begin{document}
\bstctlcite{bstctl:force_etal}

\abovedisplayshortskip=1pt
\belowdisplayshortskip=1pt
\abovedisplayskip=1pt
\belowdisplayskip=1pt
\textfloatsep=1pt
\floatsep=1pt
\intextsep=1pt

\setcounter{figure}{0}
\renewcommand{\figurename}{Fig.}
\renewcommand{\thefigure}{\arabic{figure}}
\title{Constellation-Independent Range Estimation in Payload-Based OFDM-ISAC}
\newcommand{\yang}[1]{\textcolor{blue}{\textbf{[DIY: #1]}}}
\newcommand{\han}[1]{\textcolor{magenta}{\textbf{[KWH: #1]}}}
\newcommand{\revtext}[1]{\textcolor{blue}{#1}}

\author{Dongil~Yang,~\IEEEmembership{Student Member,~IEEE,} Kaitao~Meng,~\IEEEmembership{Member,~IEEE,}    
Christos~Masouros,~\IEEEmembership{Fellow,~IEEE} and
Kawon~Han,~\IEEEmembership{Member,~IEEE}}

\begin{comment}
\thanks{
D. Yang and K. Han are with the Department of Electrical Engineering, Ulsan National Institute of Science and Technology (UNIST), Ulsan, South Korea (emails:\{dongil1223, kawon.han\}@unist.ac.kr).
K. Meng is with the Department of Electrical and Electronic Engineering, University of Manchester, Manchester, UK.
C. Masouros is with the Department of Electronic and Electrical Engineering, University College London, London, UK.}} 
\end{comment}

\maketitle

\begin{abstract}
Orthogonal frequency division multiplexing (OFDM) is a key waveform for integrated sensing and communication (ISAC) due to its spectral efficiency and compatibility with modern wireless standards. In multi-target and clutter-rich environments, however, payload-based OFDM-ISAC can suffer from data-dependent sidelobes induced by non-constant-modulus modulation symbols. To overcome these limitations, this paper proposes a region-of-interest mismatched filter (ROI-MMF) that suppresses sidelobes within a prescribed delay region while preserving the mainlobe response. By leveraging the Woodbury identity, the proposed design admits an efficient closed-form implementation whose complexity scales with the ROI size rather than the number of subcarriers. We theoretically provide the ranging mean-square error (MSE) of the designed ROI-MMF, which shows the superior performance compared to conventional matched filtering (MF) and reciprocal filtering (RF) sensing receivers. Simulations across various constellations show that the proposed sensing receiver achieves a ranging MSE approaching the Cram\'er--Rao bound (CRB), which notably confirms that our design preserves the target ranging performance even under the non-constant-modulus constellation. Finally, the framework is experimentally validated with our over-the-air OFDM-ISAC testbed.
\end{abstract}

\begin{IEEEkeywords}
Cram\'er--Rao bound (CRB), ISAC, mean-square error (MSE), mismatched filtering, OFDM, range estimation.
\end{IEEEkeywords}

\IEEEpeerreviewmaketitle

\vspace{-1\baselineskip}
\section{Introduction}
\IEEEPARstart{I}ntegrated sensing and communication (ISAC) is a key enabling technology for future wireless systems, because communication and sensing can share the same spectrum, hardware, and signaling resources. Among candidate waveforms, orthogonal frequency division multiplexing (OFDM) is particularly attractive since it offers high spectral efficiency, is already adopted in modern communication standards \cite{prasad2004ofdm}, and remains well suited for joint communication and ranging due to its favorable sidelobe behavior under cyclic-prefix (CP) signaling \cite{liu2025cp}. Thus, OFDM-ISAC can add sensing capability to wireless infrastructure without requiring major changes at the physical layer, and is being actively explored across diverse deployment scenarios and system designs \cite{xie2026ris, xie2026lightweight}.

In OFDM-ISAC, one promising direction is to use communication data payloads directly for sensing. Since payload-based sensing exploits the full OFDM frame instead of relying only on dedicated reference signals such as preambles and pilots, it improves resource efficiency, coherent processing gain, range and Doppler ambiguities, resulting in better sensing accuracy \cite{sturm2011waveform}. The ranging performance is affected by the modulation constellation, especially under non-unit-amplitude signaling in multi-target and clutter-rich environments.

For payload-based OFDM-ISAC, sensing receive filters including the matched filter (MF) and the mismatched filter (MMF) have been investigated under the impact of the modulation constellation \cite{keskin2025fundamental, han2025constellation, han2025next}. The MF preserves the original signal structure, but under non-constant-modulus constellations it produces data-dependent sidelobes\cite{liu2025uncovering}. These sidelobes increase inter-target interference and degrade delay estimation accuracy. On the other hand, the reciprocal filter (RF) and the Wiener filter (WF), both of which fall into the family of the MMF receivers, suppress this sidelobe at the cost of the signal-to-noise (SNR) loss due to amplified noise with non-constant-modulus constellations \cite{keskin2025fundamental}. Recent work in \cite{han2025constellation} analyzed these effects in closed form: the MF ranging error is determined by the fourth-order moment of the constellation, while the RF error depends on its inverse second-order moment. In other words, neither receive filter focusing the whole delay domain approaches the Cram\'er--Rao bound (CRB) independently of the constellation once high-order modulation is used. While transmit signaling designs, such as ISAC constellation shaping \cite{han2025next, geiger2026constellation, du2026probabilistic, du2024reshaping} and constellation selection \cite{meng2026constellation} may be adopted to overcome those limitations, they still introduce a performance trade-off between sensing and communications.

\begin{figure}[t!]
    \centering
    \includegraphics[width=0.9\columnwidth]{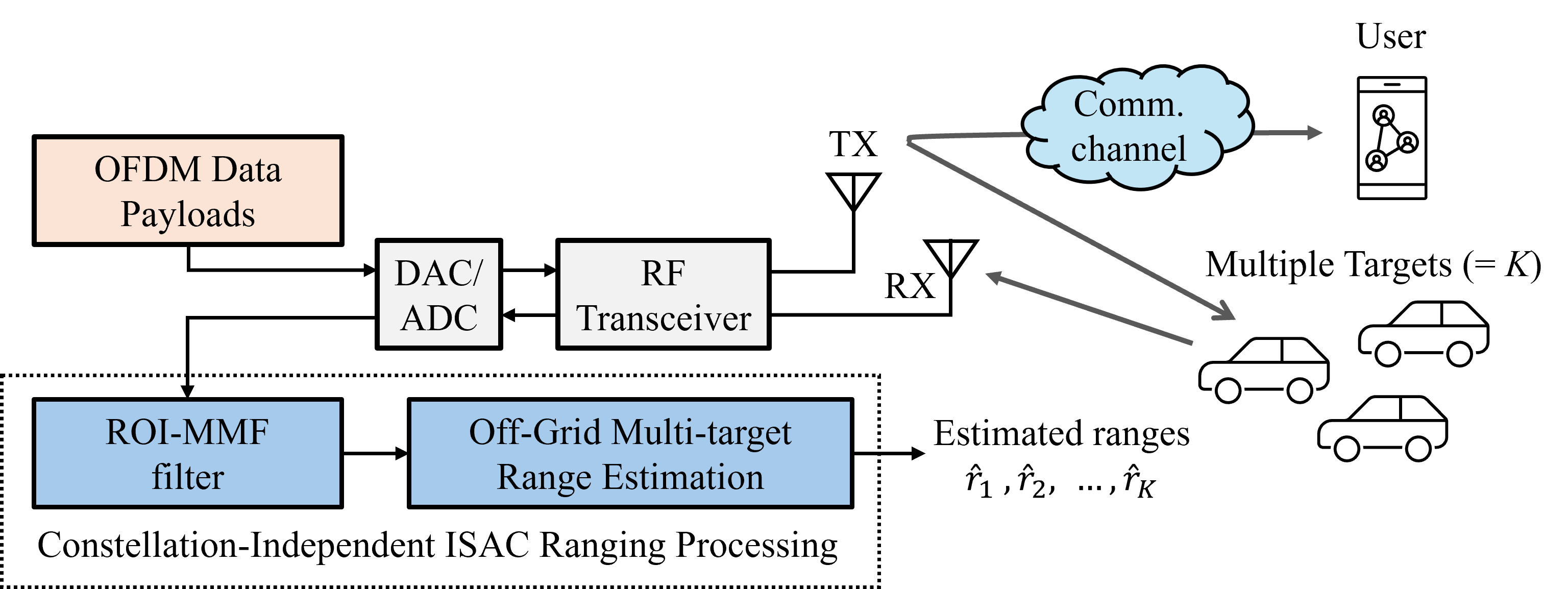}
    \caption{The constellation-independent OFDM-ISAC ranging processing, consisting of the ROI-MMF receiver and off-grid multi-target range estimation.}
    \label{fig:pipeline}
\end{figure}

To address these limitations, we design a region-of-interest mismatched filter (ROI-MMF) for payload-based OFDM-ISAC ranging as shown in Fig.~\ref{fig:pipeline}. Unlike the conventional MMF, such as RF and WF, suppressing unwanted responses across the whole delay domain, ROI-MMF suppresses sidelobes only within a prescribed ROI where target detection and estimation matter most, while preserving the mainlobe response so that the range estimation performance is maintained. This viewpoint is natural in practical sensing scenarios, since the target delay range is either known a priori or bounded by geometry in many deployments. By confining the suppression to the ROI, the proposed receiver mitigates the data-dependent sidelobe of MF without incurring the severe inverse-symbol-power noise enhancement of RF, thereby enabling constellation-independent ranging behavior. 

The main contributions of this paper are summarized as follows. First, we propose an ROI-based mismatched filtering framework for payload-based OFDM-ISAC that suppresses sidelobes over a prescribed delay region while preserving the mainlobe response. Second, we derive a closed-form solution via the Woodbury identity, whose complexity scales with the ROI size rather than the total number of subcarriers. Third, we develop a theoretical analysis under a multi-target environment and show that the proposed receiver approaches the CRB independently of the constellation, unlike MF and RF which suffer constellation-dependent penalties for target range estimation. Finally, we validate the proposed framework using numerical simulations across multiple constellations and measured OFDM-ISAC data.

\vspace{-1\baselineskip}
\section{System Model}\label{Sec::2}

\subsection{OFDM-ISAC Transmit Signal}

We consider a monostatic OFDM-ISAC transceiver employing an
OFDM waveform with $N$ subcarriers. Since the sensing receiver is co-located with the transmitter and processes the echoes of its own transmitted signal, the transmit symbols $\mathbf{x}$ are intrinsically known at the sensing receiver. The subcarrier spacing is
$\Delta f$, the total bandwidth is $B = N\Delta f$, and a cyclic
prefix (CP) of duration $T_\mathrm{cp}$ is prepended to each
OFDM symbol. The transmit vector is denoted as $\mathbf{x} =
[x_0, x_1, \ldots, x_{N-1}]^T \in \mathbb{C}^{N}$, whose entries
are i.i.d.\ communication symbols drawn from a Q-ary constellation
$\mathcal{C} = \{s_1, s_2, \ldots, s_Q\}$. Without loss of
generality, we assume the constellation is normalized to zero
mean and unit average power, i.e., $\mathbb{E}[x_n] = 0$ and
$\mathbb{E}[|x_n|^{2}] = 1$. In addition, each OFDM block is normalized such that $\|\mathbf{x}\|^2=N$.

For the sensing performance analysis in Section~\ref{Sec::3},
we define two statistical moments of the constellation as
$\mu_{4} = \mathbb{E}\!\left[|x_n|^{4}\right]$ and $\nu_{-2} = \mathbb{E}\!\left[|x_n|^{-2}\right]$, which represent the fourth-order moment and the inverse second-order
moment of the constellation, respectively \cite{han2025constellation}. For unit-modulus constellations
such as PSK, both moments reduce to $\mu_{4} = \nu_{-2} = 1$. 

\vspace{-1\baselineskip}
\subsection{Sensing Signal Model}
We consider $K$ scatterers within the CP-induced unambiguous
delay window. The scatterers are assumed to be sufficiently separated in the delay domain to be resolvable. The $k$-th scatterer is characterized by a complex
amplitude $\alpha_k \in \mathbb{C}$ and a round-trip delay
$\tau_k$, where $\alpha_k$ absorbs the path loss
and the radar cross-section of the scatterer.
The Doppler shift within one coherent processing interval (CPI)
is assumed to be negligible relative to the subcarrier spacing,
so that inter-carrier interference is avoided.

After CP removal and $N$-point FFT at the sensing receiver, the
frequency-domain receive vector is modeled as
\begin{equation}
    \mathbf{y} \;=\; \mathbf{a}^{T}\mathbf{H}\mathbf{X} + \mathbf{n},
    \label{eq:rx-matrix}
\end{equation}
where $\mathbf{a} = [\alpha_1, \ldots, \alpha_K]^{T} \in \mathbb{C}^{K}$
is the complex amplitude vector, $\mathbf{X} = \mathrm{diag}(\mathbf{x})
\in \mathbb{C}^{N \times N}$ is the diagonal transmit symbol matrix,
and $\mathbf{n} \sim \mathcal{CN}(\mathbf{0}, \sigma^{2}\mathbf{I}_{N})$
is additive white Gaussian noise. The delay-channel matrix
$\mathbf{H} = [\mathbf{h}(\tau_{1}), \ldots, \mathbf{h}(\tau_{K})]^{T}
\in \mathbb{C}^{K \times N}$ is composed of the delay steering vectors
\begin{equation}
    \mathbf{h}(\tau) \;=\; \bigl[1,\;
        e^{-j 2\pi \Delta f\, \tau},\;\ldots,\;
        e^{-j 2\pi (N-1)\Delta f\, \tau}\bigr]^{T}.
\end{equation}
The sensing SNR of the $k$-th scatterer is defined as
$\mathrm{SNR}_k = |\alpha_k|^{2}/\sigma^{2}$. To process the sensing received signal, we apply a receive filter $\mathbf{w} \in \mathbb{C}^{N}$ to obtain $\mathbf{y}_{\mathrm{out}} = \mathbf{w}^{*} \odot \mathbf{y}$.

\section{Proposed ROI-MMF Design for OFDM-ISAC}\label{Sec::3}

\subsection{ROI-MMF Design}
In many practical sensing scenarios, the target delays of interest are confined within a dedicated range, of which a bounded region can be determined by the deployment geometry. We therefore assume that all scatterers of
interest are confined to a prescribed region of interest (ROI) in
the delay domain,
\begin{equation}
    \tau_{k} \in \mathcal{T} =  [\tau_{\min},\, \tau_{\max}],
    \qquad k = 1, \ldots, K,
    \label{eq:roi-def}
\end{equation}
with known boundaries $\tau_{\min}$ and $\tau_{\max}$. In practical deployments, the ROI $\mathcal{T}$ can be determined by the link budget and deployment geometry, which bound the effective sensing range, and is thus a physically meaningful prior. When the target region cannot be determined a priori, it can be identified by an initial full-domain search and then refined by the ROI-MMF, following the conventional radar search-and-track procedure. Given the
transmit block $\mathbf{X}$ and the receive block $\mathbf{Y}$,
our objective is to jointly estimate the number of scatterers $K$
and the off-grid delay set $\{\tau_{k}\}_{k=1}^{K}$, while treating
the complex amplitudes $\{\alpha_{k}\}$ as nuisance parameters.

To suppress data-dependent sidelobes within the ROI while preserving
the mainlobe response, we design a mismatched filter $\mathbf{w} \in
\mathbb{C}^{N}$ that minimizes the sidelobe energy across a finite
set of delay shifts covering $\mathcal{T}$, subject to a unit-gain
constraint on the transmit vector $\mathbf{x}$.

Let $M_{\max}$ denote the largest delay shift (in bin units)
corresponding to $\mathcal{T}$. We denote the bilateral shift set as
\begin{equation}
    \mathcal{S} = \{-M_{\max}, \ldots, -1\} \cup \{1, \ldots, M_{\max}\},
    \label{eq:shift-set}
\end{equation}
with cardinality $|\mathcal{S}| = 2M_{\max}$. For each $m \in
\mathcal{S}$, we define the shifted steering vector
$[\mathbf{e}_{m}]_{n} = e^{j 2\pi m n/N}$ and the shifted transmit
replica $\mathbf{z}_{m} = \mathbf{x} \odot \mathbf{e}_{m} \in
\mathbb{C}^{N}$.
Multiplication by $\mathbf{e}_{m}$ in the frequency domain corresponds to a circular delay shift by $m$ bins in the time domain, so $\mathbf{z}_{m}$ represents the frequency-domain response that the receiver would observe from a target at delay 
shift $m$ relative to the mainlobe. 
Stacking the replicas yields the off-target steering matrix $\mathbf{Z} = [\,\mathbf{z}_{1}\,, \mathbf{z}_{2}\,, \;\ldots, \mathbf{z}_{2M_{\max}\,}]
\in \mathbb{C}^{N \times |\mathcal{S}|}$.
Consequently, $\|\mathbf{Z}^{H}\mathbf{w}\|_{2}^{2} = \sum_{m \in \mathcal{S}}|\mathbf{w}^{H}\mathbf{z}_{m}|^{2}$ measures the total sidelobe energy across the prescribed off-target shift set.

With this ROI in hand, the ROI-MMF design problem is formulated as
\begin{equation}
    \min_{\mathbf{w}}\;
    \|\mathbf{Z}^{H}\mathbf{w}\|_{2}^{2}
    + \lambda \|\mathbf{w}\|_{2}^{2}
    \quad \mathrm{s.t.} \quad
    \mathbf{w}^{H}\mathbf{x} = N,
    \label{eq:roi-mmf-opt}
\end{equation}
where the first term penalizes the sidelobe energy within $\mathcal{S}$,
the second term is a ridge penalty that controls noise enhancement,
and the constraint fixes the mainlobe gain. The regularization
parameter $\lambda > 0$ governs the trade-off between sidelobe
suppression and noise enhancement. Notably, the formulated problem
in \eqref{eq:roi-mmf-opt} can have a closed-form solution, which
is efficiently obtained by utilizing the Woodbury identity as
follows.

\begin{theorem}
\label{thm:woodbury}
The unique solution to \eqref{eq:roi-mmf-opt} is given by
\begin{equation}
    \mathbf{w}_{\mathrm{ROI}} =
    \frac{N}{\tilde{\mathbf{w}}^{H}\mathbf{x}}\,\tilde{\mathbf{w}},
    \quad
    \tilde{\mathbf{w}} = \frac{1}{\lambda}\!\left(
        \mathbf{x} - \mathbf{Z}
        (\lambda \mathbf{I}_{|\mathcal{S}|} + \mathbf{Z}^{H}\mathbf{Z})^{-1}
        \mathbf{Z}^{H}\mathbf{x}
    \right),
    \label{eq:w-roi}
\end{equation}
which requires only the inversion of an $|\mathcal{S}| \times
|\mathcal{S}|$ matrix.
\end{theorem}
\renewcommand\qedsymbol{$\blacksquare$}
\begin{proof}
The KKT stationarity condition of \eqref{eq:roi-mmf-opt} yields
$(\mathbf{Z}\mathbf{Z}^{H} + \lambda \mathbf{I}_{N})\mathbf{w}
= \beta \mathbf{x}$, which requires inverting an $N \times N$
matrix. Applying the Woodbury identity
\begin{equation}
    (\lambda \mathbf{I}_{N} + \mathbf{Z}\mathbf{Z}^{H})^{-1}
    = \tfrac{1}{\lambda}\mathbf{I}_{N}
    - \tfrac{1}{\lambda}\mathbf{Z}(\lambda \mathbf{I}_{|\mathcal{S}|}
    + \mathbf{Z}^{H}\mathbf{Z})^{-1}\mathbf{Z}^{H}
\end{equation}
reduces the inversion to an $|\mathcal{S}| \times |\mathcal{S}|$
system. Enforcing the gain constraint $\mathbf{w}^{H}\mathbf{x}
= N$ yields \eqref{eq:w-roi}.
\end{proof}
\vspace{-0.5\baselineskip}
Theorem~\ref{thm:woodbury} indicates that the computational complexity of the ROI-MMF design scales with the ROI size $|\mathcal{S}|$ rather than the total number of subcarriers $N$. For ROIs of practical interest with $|\mathcal{S}| < N$, this provides a substantial reduction compared to a direct inversion of the full $N \times N$ matrix. The conventional MF and RF filters are obtained directly as $\mathbf{w}_{\mathrm{MF}}=\mathbf{x}$ and $\mathbf{w}_{\mathrm{RF}}=1/\mathbf{x}^{*}$ and applied by element-wise operations at $O(N)$ cost. The proposed ROI-MMF instead attains near-CRB accuracy at the modest additional cost of an $O(N|\mathcal{S}|^{2}+|\mathcal{S}|^{3})$ filter computation, computed only once per known payload block, which is efficiently handled by the Woodbury identity otherwise a direct inversion of~\eqref{eq:roi-mmf-opt} would cost $O(N^{3})$.

\vspace{-1\baselineskip}
\subsection{Theoretical Performance Analysis of ROI-MMF}
\label{sec:perf-analysis}
We now analyze the delay estimation MSE of the ROI-MMF receiver and compare it against the MF and RF baselines. For a general filter $\mathbf{w}$, let $g_{n} = w_{n}^{*}x_{n}$ denote the effective product, and define the sidelobe as
\begin{equation}
    c_{m} = \frac{1}{N}\sum_{n=0}^{N-1} g_{n} e^{j 2\pi m n/N} = \frac{1}{N}\mathbf{w}^{H}\mathbf{z}_{m},
    \label{eq:cm}
\end{equation}
where the normalized filter gives $c_{0} = 1$. For $K$ targets with delays $\{\tau_{k}\}_{k=1}^{K}$, following the asymptotic estimator framework of \cite{han2025constellation}, the per-target delay estimation MSE under filter $\mathbf{w}$ admits the unified form
\begin{equation}
    \mathrm{MSE}_{k} =
    \frac{I_{k}(\mathbf{w}) + N_{k}(\mathbf{w})}
         {2(2\pi\Delta f)^{2}|\alpha_{k}|^{2}
          \bigl(\sum_{n} n^{2}\mathbb{E}[g_{n}]\bigr)^{2}},
    \label{eq:mse-general}
\end{equation}
where $I_{k}(\mathbf{w}) = \sum_{j \neq k} |\alpha_{j}|^{2}\, \mathrm{Var}(B_{k,j})$ with $B_{k,j} = \sum_{n} n g_{n}e^{j 2\pi n \xi_{k,j}}$ and $\xi_{k,j} = \Delta f (\tau_{k} - \tau_{j})$ captures the aggregate interference variance from the other $K-1$ targets, and $N_{k}(\mathbf{w}) = \sigma^{2}\mathbb{E}\left[\sum_{n} n^{2}|w_{n}|^{2}\right]$ captures the noise enhancement. The gain constraint $\mathbf{w}^{H}\mathbf{x} = N$ holds deterministically by construction and, together with the i.i.d.\ symmetry of $\{x_{n}\}$ across subcarriers, yields $\mathbb{E}[g_{n}] = 1$ for all $n$, so the denominator of \eqref{eq:mse-general} reduces to $(\sum_{n} n^{2})^{2}$.

\subsubsection{MF and RF baselines} For the matched filter $\mathbf{w}_{\mathrm{MF}} = \mathbf{x}$ and the reciprocal filter $\mathbf{w}_{\mathrm{RF}} = 1/\mathbf{x}^{*}$, specializing \eqref{eq:mse-general} recovers the closed-form expressions \cite{han2025constellation}
\begin{align}
    \mathrm{MSE}_{\mathrm{MF},k}
    &= \frac{3\bigl[(\mu_{4} - 1)\sum_{j\neq k}|\alpha_{j}|^{2}
                    + \sigma^{2}\bigr]}
            {8\pi^{2}\Delta f^{2}|\alpha_{k}|^{2}N^{3}},
    \label{eq:mse-mf}\\
    \mathrm{MSE}_{\mathrm{RF},k}
    &= \frac{3\sigma^{2}\nu_{-2}}
            {8\pi^{2}\Delta f^{2}|\alpha_{k}|^{2}N^{3}}.
    \label{eq:mse-rf}
\end{align}
The MF MSE is $\mu_{4}$-limited through data-dependent sidelobes, while the RF MSE is $\nu_{-2}$-limited through noise enhancement. Neither achieves the CRB under non-constant-modulus constellations.

\subsubsection{ROI-MMF} To analyze the ranging MSE of the proposed ROI-MMF under random communication signals, we first provide an energy identity for the ROI-MMF that bounds both its residual sidelobe energy and filter norm. These bounds will serve as the analytical anchors for the interference and noise-enhancement analysis developed throughout this subsection.

\begin{lemma}
\label{lem:energy}
The ROI-MMF solution \eqref{eq:w-roi} satisfies
\begin{equation}
    \|\mathbf{Z}^{H}\mathbf{w}_{\mathrm{ROI}}\|_{2}^{2}
    \leq \frac{\lambda N}{1 - \rho_{\mathcal{S}}},
    \quad
    \|\mathbf{w}_{\mathrm{ROI}}\|_{2}^{2}
    \leq \frac{N}{1 - \rho_{\mathcal{S}}},
    \label{eq:energy-identity}
\end{equation}
where $\rho_{\mathcal{S}} = \tfrac{1}{N}\mathbf{x}^{H}\mathbf{Z}
(\lambda \mathbf{I}_{|\mathcal{S}|} + \mathbf{Z}^{H}\mathbf{Z})^{-1}
\mathbf{Z}^{H}\mathbf{x} \in [0, 1)$.
\end{lemma}

\begin{proof}
Substituting $\mathbf{w}_{\mathrm{ROI}} = \beta(\mathbf{Z}\mathbf{Z}^{H}
+ \lambda \mathbf{I}_{N})^{-1}\mathbf{x}$ into the gain constraint
and applying the Woodbury identity yields $\beta = \lambda/(1 -
\rho_{\mathcal{S}})$. Left-multiplying the KKT equation by
$\mathbf{w}_{\mathrm{ROI}}^{H}$ and using $\mathbf{w}_{\mathrm{ROI}}^{H}
\mathbf{x} = N$ gives $\|\mathbf{Z}^{H}\mathbf{w}_{\mathrm{ROI}}\|_{2}^{2}
+ \lambda\|\mathbf{w}_{\mathrm{ROI}}\|_{2}^{2} = \beta N$, from which
\eqref{eq:energy-identity} follows.
\end{proof}
\vspace{-0.5\baselineskip}
Lemma~\ref{lem:energy} indicates that the residual sidelobe energy
and the filter norm of the ROI-MMF are jointly controlled by the
single quantity $\rho_{\mathcal{S}}$, with the regularization
parameter $\lambda$ governing the trade-off between the two.

Moreover, the following lemma characterizes the ROI-MMF noise enhancement in closed form.
\begin{lemma}
\label{lem:noise-enhancement}
Define $\delta^{(N)}_{\mathrm{ROI}} = \|\mathbf{w}_{\mathrm{ROI}}\|_{2}^{2}/N$ so that $N_k(\mathbf{w}_{\mathrm{ROI}})=\sigma^2\delta^{(N)}_{\mathrm{ROI}}\sum_{n} n^{2}$. Assuming that $\lambda\ll N$ and that the density $|\mathcal S|/N$ is small enough so that the leading term of the Neumann expansion dominates,
\begin{equation}
    \delta^{(N)}_{\mathrm{ROI}}
    \;=\; \frac{N}{N - \kappa_{2}\,|\mathcal{S}|},
    \label{eq:delta-N-LO}
\end{equation}
where $\kappa_{2} = \mathbb{E}\bigl[(|x_{n}|^{2}-1)^{2}\bigr]$ is the centered second moment of the constellation. 
\end{lemma}

\begin{proof}
With $\mathbf{X} = \mathrm{diag}(\mathbf{x})$ and $\mathbf{V} = [\mathbf{e}_{m_{1}}, \ldots, \mathbf{e}_{m_{|\mathcal{S}|}}]$, we have $\mathbf{Z} = \mathbf{X}\mathbf{V}$. With $u_{n} = |x_{n}|^{2} - 1$ and using the Vandermonde orthogonality $\mathbf{V}^{H}\mathbf{V} = N\mathbf{I}_{|\mathcal{S}|}$, we obtain
\begin{equation}
    \mathbf{Z}^{H}\mathbf{Z}
    = N\mathbf{I}_{|\mathcal{S}|} + \boldsymbol{\varepsilon},
    \quad
    \boldsymbol{\varepsilon} =
    \mathbf{V}^{H}\mathrm{diag}(u_{n})\mathbf{V},
    \nonumber
\end{equation}
where $\boldsymbol{\varepsilon}$ is a zero-mean Hermitian perturbation since $\mathbb{E}[u_{n}] = 0$. Setting $\xi = \lambda + N$ and $\mathbf{q} = \mathbf{Z}^{H}\mathbf{x}$, the quantity $\rho_{\mathcal{S}}$ from Lemma~\ref{lem:energy} can be rewritten as $\rho_{\mathcal{S}} = \mathbf{q}^{H}(\xi\mathbf{I} + \boldsymbol{\varepsilon})^{-1}\mathbf{q}/N$. From the energy identity \eqref{eq:energy-identity}, $\|\mathbf{w}_{\mathrm{ROI}}\|_{2}^{2} = N/(1-\rho_{\mathcal{S}}) - \|\mathbf{Z}^{H}\mathbf{w}_{\mathrm{ROI}}\|_{2}^{2}/\lambda$. The correction term $\|\mathbf{Z}^{H}\mathbf{w}_{\mathrm{ROI}}\|_{2}^{2}/\lambda$ is of order $\lambda$ and is negligible in the regime $\lambda \ll N$. Hence $\|\mathbf{w}_{\mathrm{ROI}}\|_{2}^{2} \approx N/(1-\rho_{\mathcal{S}})$, and $\delta^{(N)}_{\mathrm{ROI}} \approx 1/(1-\rho_{\mathcal{S}})$. In the regime given that $|\mathcal{S}|/N$ is small enough, the spectral-norm bound $\|\boldsymbol{\varepsilon}\|/\xi \leq C\sqrt{|\mathcal{S}|/N} < 1$ holds for a constant $C$ independent of $N$. A Neumann-series expansion of $(\xi\mathbf{I} + \boldsymbol{\varepsilon})^{-1}$ therefore converges, and the leading term gives $\rho_{\mathcal{S}} \approx \mathbf{q}^{H}\mathbf{q}/(N\xi)$. Concentration of $\mathbf{q}^{H}\mathbf{q}$ around its mean, combined with $\mathbb{E}[\mathbf{q}^{H}\mathbf{q}] = |\mathcal{S}|N\kappa_{2}$ and $\lambda \ll N$, yields $\rho_{\mathcal{S}} \approx \kappa_{2}|\mathcal{S}|/N$, leading to \eqref{eq:delta-N-LO}. 
\end{proof}
\vspace{-0.5\baselineskip}
The assumption for the density $|\mathcal{S}|/N$ is verified through Fig.~\ref{fig:deviation}, where at $|\mathcal{S}|/N = 0.4$, the gap of the estimation efficiency between theory and numerical result is only about $2\%$ and widens only as $|\mathcal{S}|/N \to 1$.

\begin{figure}[t!]
    \centering
    \includegraphics[width=1.00\columnwidth]{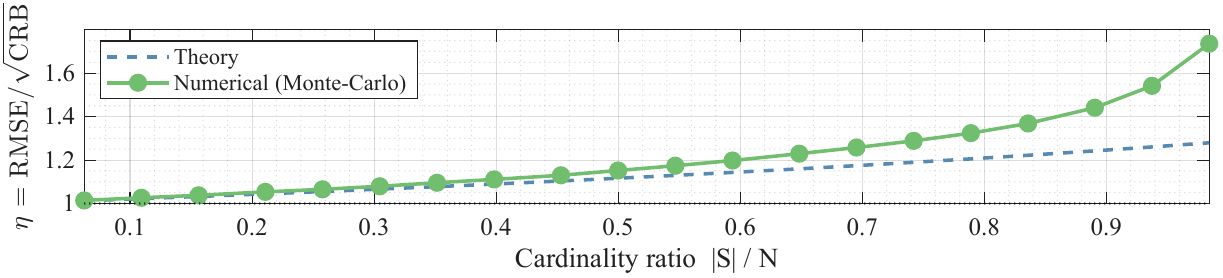}
    \caption{Estimation efficiency as a function of $|\mathcal{S}|/N$ for $256$-QAM at $N=256$, $\lambda=0.1$, $\mathrm{SNR}=10$~dB.}
    \label{fig:deviation}
\end{figure}

\noindent\textbf{Remark 1.} The noise enhancement of ROI-MMF depends on the constellation only through the centered moment $\kappa_{2}$ of $|x_{n}|^{2}$, which remains bounded for QAM formats ($\kappa_{2} \to 0.40$ as $M \to \infty$). This is in sharp contrast
to the RF receiver, whose noise penalty $\nu_{-2}$ grows unboundedly
as constellation points approach the origin.

Building on Lemma 1 and Lemma 2, we are now ready to derive a closed-form expression for the MSE of the ROI-MMF as the following theorem states. 
\vspace{-0.5\baselineskip}
\begin{theorem}[ROI-MMF delay estimation MSE]
\label{thm:roi-mmf-mse}
The closed-form expression of the ROI-MMF delay estimation MSE is approximately given by
\begin{equation}
    \mathrm{MSE}_{\mathrm{ROI},k}
    = \frac{3\sigma^{2}}{8\pi^{2}\Delta f^{2}|\alpha_{k}|^{2}N^{3}}
      \cdot \delta_{\mathrm{ROI}},
    \label{eq:mse-roi}
\end{equation}
with $\delta_{\mathrm{ROI}} \ = \delta^{(N)}_{\mathrm{ROI}}
+ \delta^{(I)}_{\mathrm{ROI}}$.
\end{theorem}

\begin{proof}
Let $\rho_{\mathrm{tot}} = \sum_{j\neq k}|\alpha_{j}|^{2}/\sigma^{2}$ denote the aggregate interference SNR. The interference residual of the ROI-MMF, defined as $\delta^{(I)}_{\mathrm{ROI}} = \sum_{j\neq k}|\alpha_{j}|^{2}\, \mathrm{Var}(B_{k,j}^{\mathrm{ROI}})/(\sigma^{2}\sum_{n} n^{2})$ such that $I_k(\mathbf{w}_{\mathrm{ROI}})=\sigma^2\delta^{(I)}_{\mathrm{ROI}}\sum_{n} n^{2}$, satisfies
\begin{equation}
    \delta^{(I)}_{\mathrm{ROI}}
    \leq \frac{\rho_{\mathrm{tot}}\,\lambda}{N(1-\rho_{\mathcal{S}})},
    \label{eq:delta-I-bound}
\end{equation}
which follows from the suppression property $\sum_{m \in \mathcal{S}}|c_{m}|^{2} \leq \lambda/(N(1 - \rho_{\mathcal{S}}))$ established in Lemma~\ref{lem:energy} together
with an application of the Cauchy--Schwarz inequality. The bound \eqref{eq:delta-I-bound} is structurally controlled by the design parameter $\lambda$ and vanishes as $\lambda \to 0$. Combining this interference residual with the noise enhancement $\delta^{(N)}_{\mathrm{ROI}}$ characterized in Lemma~\ref{lem:noise-enhancement} and substituting into the unified form \eqref{eq:mse-general} yields \eqref{eq:mse-roi}.
\end{proof}

\vspace{-0.5\baselineskip}
\noindent\textbf{Remark 2.} Although the bound \eqref{eq:delta-I-bound} vanishes as $\lambda\to0$, the noise enhancement grows in that same limit, so retaining the $\lambda$-dependence of both terms reveals a finite optimum. Setting $\mathrm{d}\delta_{\mathrm{ROI}}/\mathrm{d}\lambda = 0$ in the regime $\lambda \ll N$ yields a unique optimum
\begin{equation}
    \lambda^{\star} = \frac{N(N-\kappa_{2}|\mathcal{S}|)}{2\rho_{\mathrm{tot}}(N-\kappa_{2}|\mathcal{S}|)-N},
    \label{eq:lambda-opt}
\end{equation}
so the optimal regularization is inversely proportional to the target SNR. Nevertheless, since $\rho_{\mathrm{tot}}$ is not known a priori in practice, we adopt the fixed value $\lambda = 0.1$, which empirically incurs a negligible loss relative to \eqref{eq:lambda-opt}.

\subsubsection{Performance Comparisons} Applying the standard Fisher information analysis yields the per-target CRB as $\mathrm{CRB} = \frac{\sigma^{2}}{8\pi^{2}\Delta f^{2}|\alpha_{k}|^{2}\sum_{n} n^{2}}$\cite{han2025constellation}. To compare the ranging performance of MF, RF, and ROI-MMF, we introduce the estimation efficiency defined as $\eta = \mathrm{RMSE}/\mathrm{\sqrt{CRB}}$. Then, dividing \eqref{eq:mse-mf}, \eqref{eq:mse-rf}, and \eqref{eq:mse-roi} yields
\begin{equation}
    \eta_{\mathrm{MF}} = \sqrt{1 + (\mu_{4} - 1)\rho_{\mathrm{tot}}},\;
    \eta_{\mathrm{RF}} = \sqrt{\nu_{-2}},\;
    \eta_{\mathrm{ROI}} = \sqrt{\delta_{\mathrm{ROI}}},
    \label{eq:etas}
\end{equation}
where $\rho_{\mathrm{tot}} = \sum_{j\neq k}|\alpha_{j}|^{2}/\sigma^{2}$
is the aggregate interference SNR introduced above. Unlike MF and RF whose estimation efficiency scales according to $\mu_4$ and $\nu_{-2}$ of the constellation, the ROI-MMF shows the superior performance less affected by the modulation constellation geometry, which is shown in the following corollary.

\begin{figure}[!t]
    \centering
    \subfloat[16QAM]{%
        \includegraphics[width=0.5\columnwidth]{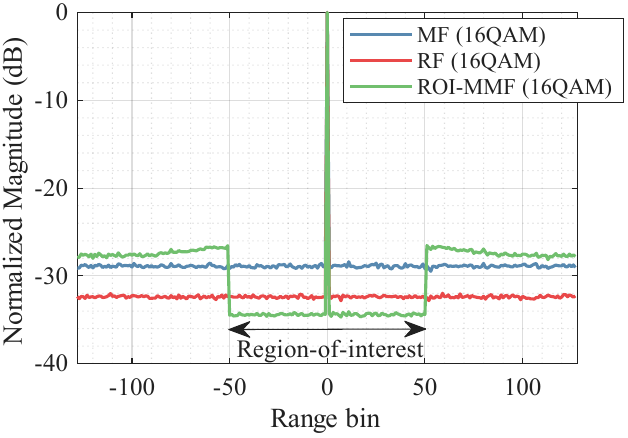}%
        \label{fig:profile_16qam}}%
    \hspace{0.00\columnwidth}%
    \subfloat[256QAM]{%
        \includegraphics[width=0.5\columnwidth]{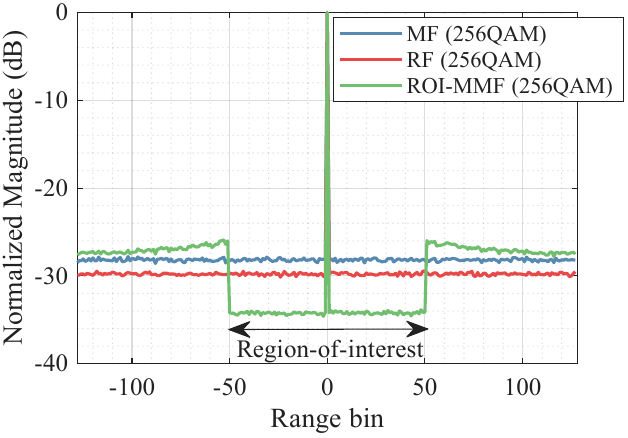}%
        \label{fig:profile_256qam}}%
    \caption{Normalized range profiles of the MF, RF, and the proposed ROI-MMF receivers under (a) 16-QAM and (b) 256-QAM at a target SNR of 10 dB.}
    \label{fig:SNR10dB}
\end{figure}

\begin{corollary}
\label{cor:dominance}
In the regime $\lambda \ll N$ and $|\mathcal{S}| \ll N$, the
ROI-MMF estimation efficiency $\eta_{\mathrm{ROI}}$ nearly approaches unity
independently of the constellation geometry,
\end{corollary}

\begin{proof}
From \eqref{eq:etas}, $\eta_{\mathrm{ROI}}^{2} =
\delta_{\mathrm{ROI}} = \delta^{(N)}_{\mathrm{ROI}} +
\delta^{(I)}_{\mathrm{ROI}}$. By Lemma~\ref{lem:noise-enhancement}
and \eqref{eq:delta-I-bound}, $\delta^{(N)}_{\mathrm{ROI}} \approx
1 + \kappa_{2}|\mathcal{S}|/N$ and $\delta^{(I)}_{\mathrm{ROI}}
\leq \rho_{\mathrm{tot}}\lambda/N$, so $\delta_{\mathrm{ROI}}$
deviates from unity only by $\kappa_{2}|\mathcal{S}|/N +
\rho_{\mathrm{tot}}\lambda/N$ in the stated regime. As the
centered moment $\kappa_{2}$ saturates to a small constant,
$\delta_{\mathrm{ROI}}$ remains close to unity and
constellation-independent.
\end{proof}

\vspace{-0.5\baselineskip}
\noindent\textbf{Remark 3.} The constellation-independent property of ROI-MMF stems from the fact that $\delta_{\mathrm{ROI}}$ depends on the constellation only through the centered moment $\kappa_{2}$ multiplied by $|\mathcal{S}| / N \ll 1$, and through the factor $\rho_{\mathrm{tot}}\lambda/N$ which is structurally controlled by
$\lambda$. In contrast, the MF penalty grows with $\mu_{4}$ and the RF penalty grows with $\nu_{-2}$. The advantage of ROI-MMF over the classical receive filters therefore widens as the constellation moves away from constant
modulus.

\vspace{-0.5\baselineskip}
\section{Numerical Simulations and Experimentation}
\label{Sec::4}
\subsection{Numerical Simulation Results}
For the numerical simulations, we adopt a CP-OFDM waveform with
$N = 256$, $B = 50$~MHz, and a CP length of $0.64~\mu$s. We coherently process $L$ number of OFDM symbols in one CPI, of which value is $L=8$.  The proposed ROI-MMF is configured with the bilateral shift set $\mathcal{S} = \{-50, \ldots, -1\} \cup \{1, \ldots, 50\}$, corresponding to the range interval of $-150$~m to $150$~m, and regularization parameter $\lambda = 0.1$. Each result is averaged over $2{,}000$ Monte-Carlo runs. A subspace-based matrix pencil (MP) estimator is employed for off-grid range estimation.

We first examine the structural difference between the three receivers through their range profiles. Fig.~\ref{fig:SNR10dB} shows the normalized range profiles of MF, RF, and ROI-MMF under $16$-QAM and $256$-QAM at a target SNR of $10$~dB with a single scatterer. The MF profile exhibits a constellation-dependent sidelobe floor, since its sidelobe energy is governed by the fourth moment $\mu_{4}$. The RF profile suppresses the sidelobes but lifts the noise floor. The proposed ROI-MMF, in contrast, forms a clear notch across the prescribed ROI for both constellations, and the two profiles nearly overlap, which confirms the constellation-independent suppression behavior predicted by Lemma~\ref{lem:noise-enhancement}.

\begin{comment}
        \subfloat[QPSK]{%
        \includegraphics[width=0.5\columnwidth]{Figure_3_QPSK.pdf}%
        \label{fig:rmse_qpsk}%
    }%
    \hspace{0.00\columnwidth}%
    \subfloat[16QAM]{%
        \includegraphics[width=0.5\columnwidth]{Figure_3_16QAM.pdf}%
        \label{fig:rmse_16qam}%
    }%
    \vspace{-0.5em}
\end{comment}
\begin{figure}[t!]
    \centering
    \subfloat[16QAM]{%
        \includegraphics[width=0.5\columnwidth]{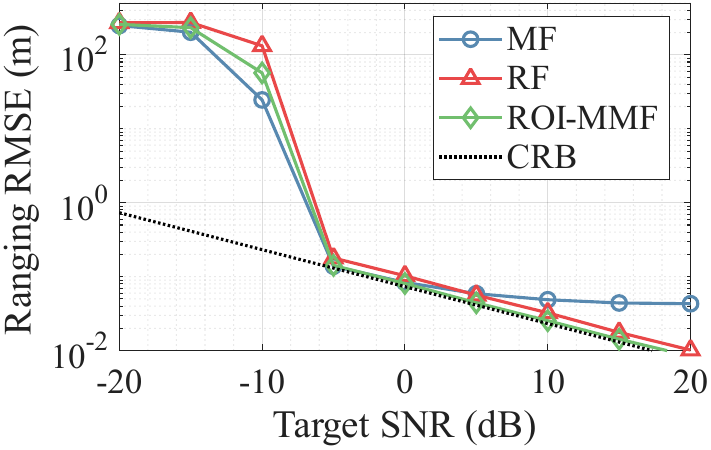}%
        \label{fig:rmse_16qam}%
    }%
    \hspace{0.00\columnwidth}%
    \subfloat[256QAM]{%
        \includegraphics[width=0.5\columnwidth]{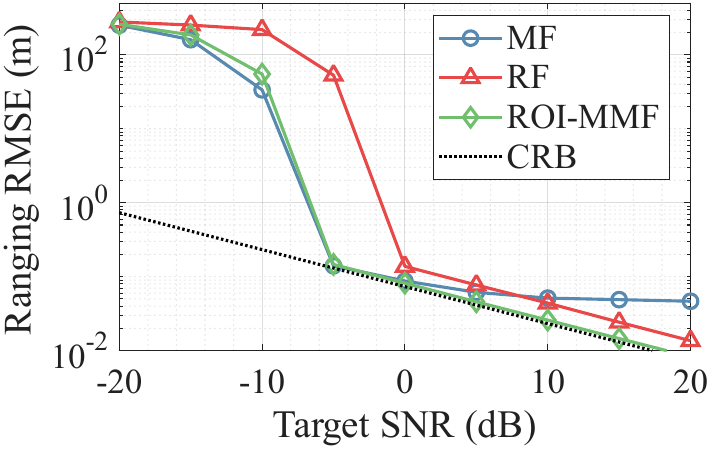}%
        \label{fig:rmse_256qam}%
    }%
    \vspace{-0.5em}
    \caption{Range estimation RMSE versus target SNR for the MF, RF, and the proposed ROI-MMF under (a) 16-QAM, and (b) 256-QAM modulation. $N = 256$, $L = 8$, 
$\lambda = 0.1$, and $\mathcal{S} = \{\pm 1, \ldots, \pm 50\}$.}
    \label{fig:rmse_vs_snr}
\end{figure}

Next, we evaluate the multi-target ranging accuracy. Two scatterers are placed at off-grid ranges $r_{1} = 60.9$~m and $r_{2} = 90.9$~m within the ROI, and $L = 8$ consecutive OFDM symbols are collected within one CPI. Fig.~\ref{fig:rmse_vs_snr} shows the range estimation RMSE versus the target SNR for the two different QAM constellations, together with the CRB. For non-constant-modulus constellations, the MF curves saturate as the SNR increases, since the sidelobe interference from the other target dominates the estimation error, consistent with the $\mu_{4}$-limited expression \eqref{eq:mse-mf}. The RF curves exhibit a constant gap from the CRB whose magnitude grows with $\nu_{-2}$, matching \eqref{eq:mse-rf}. In contrast, the proposed ROI-MMF tracks the CRB across all constellations and SNRs. Table~\ref{tab:rmse-crb} reports the estimation efficiency $\eta$ at $\mathrm{SNR} = 10$~dB, comparing the closed-form predictions \eqref{eq:etas} against the simulated values across QPSK through $256$-QAM. The closed-form expressions closely match the numerical results, providing validation of the constellation-independent behavior established in Corollary~\ref{cor:dominance}.

\begin{table}[t!]
\centering
\fontsize{7}{9}\selectfont
\renewcommand{\arraystretch}{0.9}
    \caption{The values of theoretical and numerical 
    $\eta = \mathrm{RMSE}/\mathrm{\sqrt{CRB}}$.}
    \label{tab:rmse-crb}
    \begin{tabular}{>{\centering\arraybackslash}m{4em}|
                    >{\centering\arraybackslash}m{2.5em}
                    >{\centering\arraybackslash}m{2.5em}|
                    >{\centering\arraybackslash}m{2.5em}
                    >{\centering\arraybackslash}m{2.5em}|
                    >{\centering\arraybackslash}m{2.5em}
                    >{\centering\arraybackslash}m{2.5em}}
    \toprule
    & \multicolumn{2}{c|}{MF}
    & \multicolumn{2}{c|}{RF}
    & \multicolumn{2}{c}{ROI-MMF} \\
    Const. & Theory & Num. & Theory & Num. & Theory & Num. \\
    \midrule
    QPSK    & 1.000 & 1.013 & 1.000 & 1.013 & 1.000 & 1.013 \\
    16QAM   & 2.049 & 2.102 & 1.374 & 1.398 & 1.069 & 1.095 \\
    32QAM   & 2.025 & 2.013 & 1.491 & 1.485 & 1.067 & 1.098 \\
    64QAM   & 2.193 & 2.188 & 1.638 & 1.607 & 1.084 & 1.102 \\
    128QAM  & 2.105 & 2.071 & 1.730 & 1.734 & 1.075 & 1.099 \\
    256QAM  & 2.225 & 2.212 & 1.850 & 1.894 & 1.087 & 1.115 \\
    \bottomrule
\end{tabular}
\end{table}

\vspace{-1\baselineskip}
\subsection{Experimental Results}
Beyond numerical simulations, we further validate the proposed receiver design and its theory through over-the-air experiments. To this end, a software-defined radio
(AD9363) with directional antennas of $12$~dBi gain is employed. The center frequency is set to $2.4$~GHz with a bandwidth of $B = 20$~MHz and $N = 512$ subcarriers.
In addition, $512$ OFDM symbols are transmitted per frame. The OFDM symbol duration is $14.4~\mu$s with a CP length of $1.6~\mu$s, and the sampling rate is $60$~MHz. The ROI is
set to $[50, 250]$~m to cover the deployment geometry of interest, and
each measurement result is obtained by averaging over $100$ independent
frames. As the prototype operates in a monostatic configuration with the transmitter and sensing receiver co-located on a single transceiver, no separate amplitude or phase calibration is applied. The desired target, a building wall with high reflectivity, is located approximately $130$~m from the ISAC node, as illustrated in Fig.~\ref{fig:england}.

\begin{figure}[t!]
    \centering
    \includegraphics[width=0.7\columnwidth]{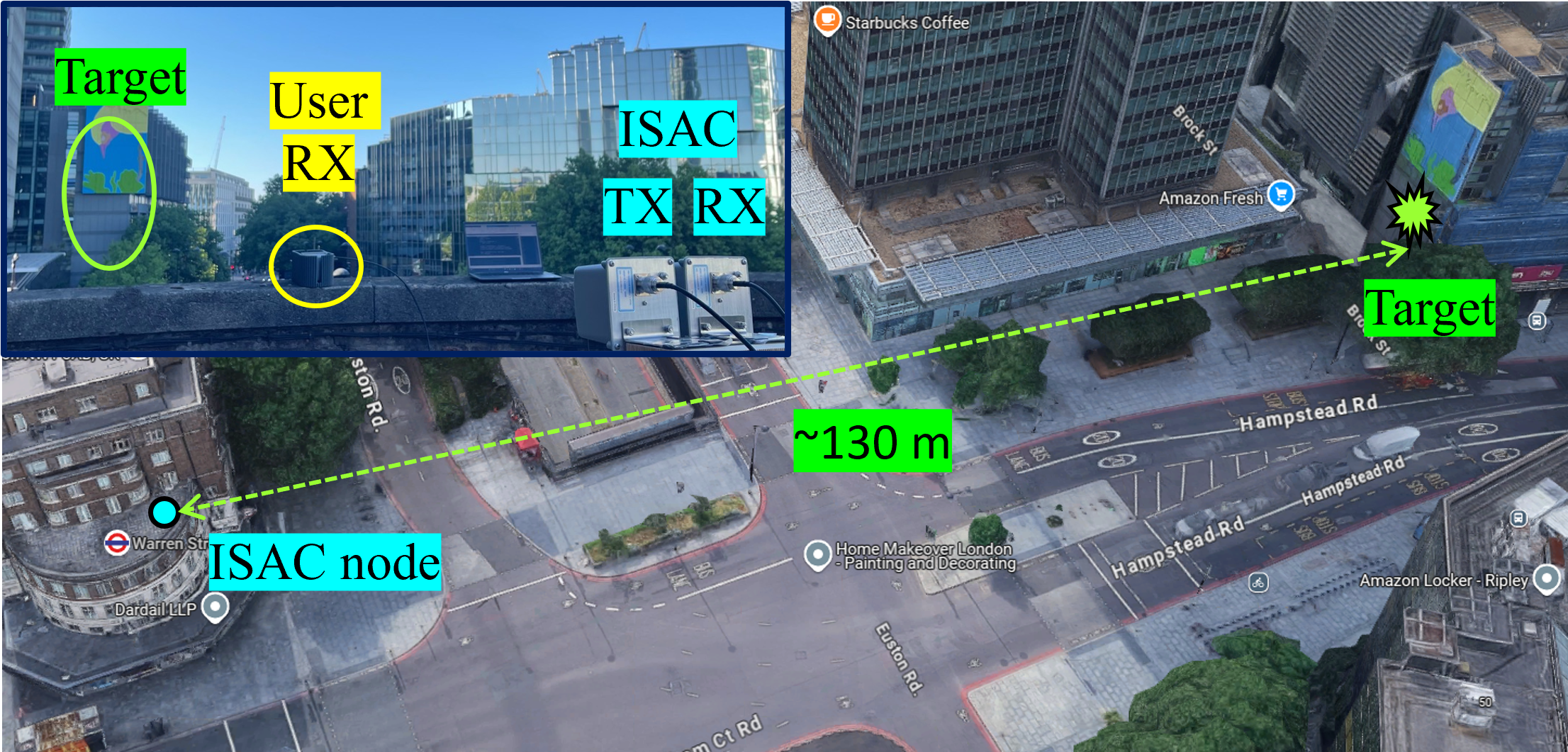}
    \caption{ Photograph of the measurement setup.}
    \label{fig:england}
\end{figure}

\begin{figure}[t!]
    \centering
    \subfloat[16QAM]{%
        \includegraphics[width=0.4\columnwidth]{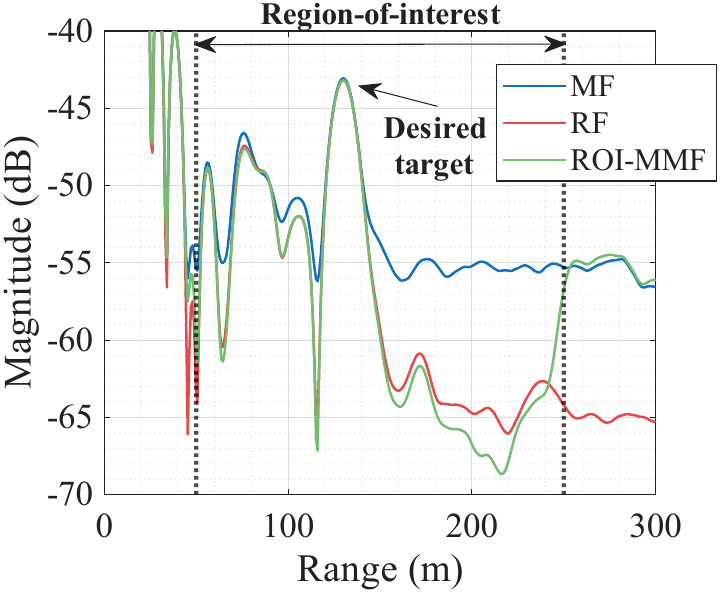}%
        \label{fig:range_16qam}}
    \subfloat[32APSK]{%
        \includegraphics[width=0.4\columnwidth]{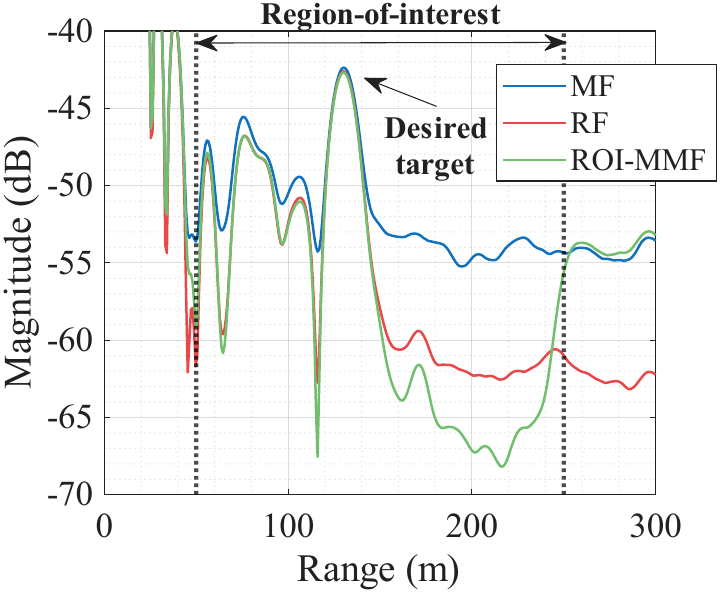}%
        \label{fig:range_32apsk}}
    \caption{Measured range profiles of the MF, RF, and ROI-MMF receivers under (a) 16-QAM and (b) customized 32-APSK modulation.}
    \label{fig:range_profile}
\end{figure}

Fig.~\ref{fig:range_profile} shows the measured range profiles
under $16$-QAM and a customized $32$-APSK modulation. As predicted by
the analysis, the MF profile suffers from elevated data-dependent
sidelobes that mask the desired target peak. The RF
profile suppresses these sidelobes but lifts the overall noise floor. The proposed ROI-MMF effectively suppresses the
response within the prescribed ROI while preserving the target peak
at $\sim\!130$~m, achieving the highest in-ROI dynamic range for both
constellations. This experimentally confirms that the
constellation-independent ranging property established theoretically also
carries over to practical over-the-air deployments.

\vspace{-0.5\baselineskip}
\section{Conclusion}
\label{Sec::5}
In this paper, we have proposed a ROI-MMF for payload-based OFDM-ISAC ranging, which suppresses data-dependent sidelobes within a prescribed delay region while
preserving the mainlobe response. Through an estimation-theoretic analysis,
we have shown that the proposed receiver dominates conventional filters and nearly approaches the CRB independently of the modulation constellation. The theoretical findings are validated through numerical simulations and over-the-air measurements, providing a constellation-independent receive-filter framework.

\vspace{-1\baselineskip}
\bibliographystyle{IEEEtran}
\bibliography{IEEEabrv,reference}
\end{document}